\documentclass[a4paper,UKenglish,cleveref, autoref, thm-restate]{lipics-v2021}

\usepackage{algorithm}
\usepackage{algorithmicx}
\usepackage{algpseudocode}
\usepackage{todonotes}
\usepackage{hyperref}
\algrenewcommand\algorithmicrequire{\textbf{Input:}}
\algrenewcommand\algorithmicensure{\textbf{Output:}}
\newtheorem*{problem}{Problem}

\hideLIPIcs  

\nolinenumbers

\title{Hamiltonian path and Hamiltonian cycle are solvable in polynomial time in graphs of bounded independence number\footnote{This paper is outdated. Please refer to the joint paper 
\href{http://arxiv.org/abs/2403.05943}{\textit{Path Cover, Hamiltonicity, and Independence Number: An FPT Perspective, CoRR, abs/2403.05943 (2025)}}, which combines results of this paper and F. V. Fomin, P. A. Golovach, D. Sagunov, and K. Simonov's work: \href{https://arxiv.org/abs/2403.05943v1}{\textit{Hamiltonicity, Path Cover, and Independence Number: An FPT Perspective, CoRR, abs/2403.05943v1 (2024}).}}} 

\titlerunning{Hamiltonian path and cycle in graphs of bounded independence number} 

\author{Nikola Jedli\v{c}kov\'{a}}{Department of Applied Mathematics, Faculty of Mathematics and Physics, Charles University, Prague, Czech Republic}{jedlickova@kam.mff.cuni.cz}{https://orcid.org/0000-0001-9518-6386}{Supported by GAUK 370122.}

\author{Jan Kratochvíl}{Department of Applied Mathematics, Faculty of Mathematics and Physics, Charles University, Prague, Czech Republic}{honza@kam.mff.cuni.cz}{https://orcid.org/0000-0002-2620-6133}{}

\authorrunning{N. Jedličková, J. Kratochvíl} 

\Copyright{Nikola Jedličková and Jan Kratochvíl} 

\ccsdesc[500]{Mathematics of computing~Graph theory} 

\keywords{Graph, Hamiltonian path, Hamiltonian cycle, Hamiltonian connectedness, Path Cover, Independence number, $L(2,1)$-labelling, Polynomial-time algorithm, NP-completeness} 

\category{} 

\relatedversion{} 

\EventEditors{John Q. Open and Joan R. Access}
\EventNoEds{2}
\EventLongTitle{42nd Conference on Very Important Topics (CVIT 2016)}
\EventShortTitle{CVIT 2016}
\EventAcronym{CVIT}
\EventYear{2016}
\EventDate{December 24--27, 2016}
\EventLocation{Little Whinging, United Kingdom}
\EventLogo{}
\SeriesVolume{42}
\ArticleNo{23}

\def\computationproblem#1#2#3{
  \begin{center}
  \begin{tabular}{rp{0.8\textwidth}}
  {\sc Problem:\enspace}&#1\\
  {\sc Input:\enspace}&#2\\
  {\sc Question:\enspace}&#3\\
  \end{tabular}
  \end{center}
}

\begin{document}

\maketitle

\begin{abstract}
A Hamiltonian path (a Hamiltonian cycle) in a graph is a path (a cycle, respectively) that traverses all of its vertices. The problems of deciding their existence in an input graph are well-known to be NP-complete. In fact, they belong to the first problems shown to be computationally hard when the theory of NP-completeness was being developed. A lot of research has been devoted to the complexity of Hamiltonian path and Hamiltonian cycle problems for special graph classes, yet only a handful of positive results are known. The complexities of both of these problems have been open for graphs of independence number at most $3$. We answer this question in the general setting of graphs of bounded independence number by showing that both these problems are in the class XP when parameterized by the independence number of the input graph.

We also consider a newly introduced problem called \emph{Hamiltonian-$\ell$-Linkage}  which is related to the notions of a path cover and of a linkage in a graph. This problem asks if $\ell$ given pairs of vertices in an input graph can be connected by disjoint paths that together traverse all vertices of the graph. For $\ell=1$, Hamiltonian-1-Linkage asks for existence of a Hamiltonian path connecting a given pair of vertices.
Our main result reads that for every pair of integers $k$ and $\ell$, the  Hamiltonian-$\ell$-Linkage problem is polynomial time solvable for graphs of independence number at most $k$. 

As an application of our result, we provide a complete characterization of the computational complexity of the $L(2,1)$-labelling problem on $H$-free graphs, and of the related $L'(2,1)$-labelling problem on $H$-free graphs for triangle-free parameter graphs $H$. 
\end{abstract}

\section{Introduction}\label{sec:Intro}

A cycle in a graph is {\em Hamiltonian} if it contains all vertices of the graph. A graph is called {\em Hamiltonian} if it contains a Hamiltonian cycle. The notion of Hamiltonian graphs is well-known and intensively studied in graph theory. Many sufficient conditions for Hamiltonicity of graphs are known (e.g., minimum degree at least $|V(G)|/2$ which goes back to Dirac~\cite{dirac1952some}, the theorem of Ore~\cite{ore1960note} whose short proof from \cite{bondy2003short} is one of the jewels of graph theory, Chv\'atal's conditions on the degree sequence~\cite{chvatal1972hamilton}), but a simple necessary and sufficient condition is not known (and not likely to exists, as the problem itself is NP-complete). 
On the other hand, many open questions and conjectures are around. The more than 50 years old conjecture of Barnette states that every cubic planar bipartite 3-connected graph is Hamiltonian, to mention at least one. 

Hamiltonian cycles in solution spaces of configurations of certain types are called {\em Gray codes}. A prominent example is that the $2^n$ bitvectors of length $n$ (i.e., the vertices of the $n$-dimensional hypercube) can be arranged in a cyclic sequence 
so that the Hamming distance (i.e., the number of coordinates in which they differ) of any two consecutive ones equals $1$. In this direction we have to mention the recent proof of the famous Middle-levels conjecture, stating that the middle level subgraph of the hypercube of odd dimension is Hamiltonian~\cite{Mutze2016}. M\"utze also proved another almost 50 years old conjecture on Hamiltonicity of odd Kneser graphs in~\cite{Mutze2018}.

A {\em Hamiltonian path} in a graph is a path that contains all vertices of the graph. Obviously, every Hamiltonian graph contains a Hamiltonian path, and a graph is Hamiltonian if it contains a Hamiltonian path connecting a pair of adjacent vertices. A graph is called {\em Hamiltonian connected} if every two distinct vertices are connected by a Hamiltonian path. Chv\'atal and Erd\H{o}s stated and proved elegant sufficient conditions for Hamiltonian connectedness and for the existence of Hamiltonian paths in terms of comparing the vertex connectivity and independence number of the graph under consideration, cf. Proposition~\ref{prop:chvatal}.
For an excellent survey on Hamiltonian graphs cf.~\cite{gould2003advances} and a more recent one~\cite{gould2014recent}.

From a computational complexity point of view, all these problems are hard. Karp~\cite{Karp1972} proved already in 1972 that deciding the existence of Hamiltonian paths and cycles in an input graph are NP-complete problems. In this sense the problem of deciding the existence of a Hamiltonian path that connects two given vertices is the canonical one - if this can be solved in polynomial time in a given graph (or graphs from a given graph class), then 
existence of a Hamiltonian path, existence of a Hamiltonian cycle, and Hamiltonian connectedness 
can be solved in polynomial time by checking all pairs of (adjacent, in the case of Hamiltonian cycle) vertices of the input graph.    

In order to study the borderline between easy (polynomial-time solvable) and hard (NP-complete) variants of the problems, researchers have intensively studied the complexity of Hamiltonian-related problems in special graph classes. 

The existence of a Hamiltonian cycle remains NP-complete on planar graphs~\cite{garey1976planar}, circle graphs~\cite{damaschke1989hamiltonian} and for several other generalizations of interval graphs~\cite{bertossi1986hamiltonian}, for chordal bipartite graphs~\cite{mu96} and for split (and therefore also for chordal) graphs~\cite{golumbic2004algorithmic}, but is solvable in linear time in interval graphs~\cite{keil1985finding} and in convex bipartite graphs~\cite{mu96}.
The existence of a Hamiltonian path can be decided in polynomial time for cocomparability graphs~\cite{damaschke1991finding} and for circular arc graphs~\cite{damaschke1993paths}. 

Many of the above mentioned graph classes are hereditary, i.e., closed in the induced subgraph order, and as such can be described by collections of forbidden induced subgraphs. This has led to carefully examining $H$-free graphs, i.e., graph classes with a single forbidden induced subgraph. For many graphs $H$, the class of $H$-free graphs has nice structural properties (e.g., $P_4$-free graphs are exactly the cographs) or their structural properties can be used to design polynomial algorithms for various graph theory problems (as a recent example, cf. colouring of $P_6$-free graphs in~\cite{Spirkl2019}). (Here and later throughout the graph, $P_k$ denotes the path on $k$ vertices, $K_k$ denotes the complete graph on $k$ vertices, and $G+H$ denotes the disjoint union of (isomorphic copies of) the graphs $G$ and $H$.) For Hamiltonian-type problems, consider three-vertex graphs $H$. For $H=K_3$, both Hamiltonian cycle and Hamiltonian path are NP-complete on triangle-free graphs, since they are NP-complete on bipartite graphs~\cite{mu96}. For $H=3K_1$, the edgeless graph on three vertices, both Hamiltonian cycle and Hamiltonian path are polynomial time solvable~\cite{Duf1981}. The remaining two graphs, $P_3$ and $K_1+K_2$, are induced subgraphs of $P_4$, thus the corresponding class of $H$-free graphs is a subclass of cographs, and as such a subclass of cocomparability graphs, in which the Hamiltonian path and Hamiltonian cycle problems are solvable in polynomial time~\cite{damaschke1991finding}. However, a complete characterization of graphs $H$ for which Hamiltonian path or Hamiltonian cycle problems are solvable in polynomial time (and for which they are NP-complete) is not in sight. Until now, the complexity has been open even for the next smallest edgeless graph $4K_1$. The main goal of this paper paper is to answer this complexity question for all edgeless forbidden induced subgraphs.

\begin{theorem}\label{thm:main}
For every integer $k$, the existence of a Hamiltonian path and of a Hamiltonian cycle can be decided in polynomial time for $kK_1$-free graphs. In other words, Hamiltonian path and Hamiltonian cycle are in the class XP when parameterized by the independence number of the input graph. 
\end{theorem} 

The approach we use is motivated by Chv\'atal-Erd\H{o}s's theorems on highly connected graphs that we restate in the following section. Informally speaking, if a $kK_1$-free graph is highly connected, the existence of a Hamiltonian path and a Hamiltonian cycle are guaranteed. If its connectivity is low, we find a small vertex cut, delete it and process the connected components of the resulting graph recursively. However, a Hamiltonian path may wind through the graph and visit each component several times. This leads us to introducing a new notion of {\em Hamiltonian-$\ell$-linkage} in a graph as a linkage consisting of disjoint paths connecting specified pairs of vertices and together using all vertices of the graph. We first prove a strengthening of Chv\'atal-Erd\H{o}s theorem for Hamiltonian-linkage, exploiting a result of Thomas and Wollan on linkages in highly connected graphs. Based on this, we prove in Theorem~\ref{thm:Ham-linkage-poly} that for all fixed $k$ and $\ell$, the Hamiltonian-$\ell$-linkage problem is polynomial time solvable in $kK_1$-free graphs. Theorem~\ref{thm:main} then follows when setting $\ell=1$. The Hamiltonian-$\ell$-linkage problem is formally described and the results are proved in Section~\ref{sec:linkage}. 

Since the exponent of the polynomial function expressing the running time of this algorithm depends on $k$ and $\ell$ and large multiplicative constants are involved, we inspected the cases of small $k$ and Hamiltonian paths and cycles in a companion paper \cite{jedlivckova2024structure}. For $k=3,4$ and $5$, there are  explicitly described situations in which a Hamiltonian path does exist. This description leads to a significant improvement of the running time estimates.        

As an application of these results, we consider the concept of distance constrained labellings, a notion stemming from the Frequency Assignment Problem, in Section~\ref{sec:L21}. We give a complete characterization (with respect to $H$) of the computational complexity of $L(2,1)$-labelling on $H$-free graphs, and a characterization of the computational complexity of $L'(2,1)$-labelling on $H$-free graphs for triangle-free graphs $H$. 
Finally, we gather concluding remarks and open problems in Section~\ref{sec:Conclusion}.

\section{Preliminaries}

We consider simple undirected graphs without loops or multiple edges.  
The vertex set of a graph $G$ is denoted by $V(G)$, its edge set by $E(G)$. Edges are considered as two-element sets of vertices, thus we write $u\in e$ to express that a vertex $u$ is incident to an edge $e$.  For the sake of brevity, we write $uv$ instead of $\{u,v\}$ for the edge containing vertices $u$ and $v$. We say that $u$ is {\em adjacent to} $v$ if $uv\in E(G)$. The {\em degree} of a vertex is the number of other vertices adjacent to it. The subgraph of $G$ induced by vertices $A\subseteq V(G)$ will be denoted by $G[A]$. The {\em independence number}  of a graph $G$, denoted by $\alpha(G)$, is the order of the largest edgeless induced subgraph of $G$. With the standard notion of $K_k$ being the complete graph with $k$ vertices and $G+H$ being the disjoint union of graphs $G$ and $H$, $\alpha(G)$ is equal to the largest $k$ such that $kK_1$ is an induced subgraph of $G$. A graph is called {\em $H$-free} if it contains no induced subgraph isomorphic to $H$.  

A {\em path} in a graph $G$ is a sequence of distinct vertices such that any two consecutive ones are adjacent. The {\em length} of a path is the number of its edges. A {\em cycle} is formed by a path of length greater than 1 which connects two adjacent vertices. The path (cycle) is {\em Hamiltonian} if it contains all vertices of the graph. The \emph{path cover number} of $G$, denoted by $\mathrm{pc}(G)$, is the smallest number of vertex disjoint paths that cover all vertices of $G$.

A graph is connected if any two vertices are connected by a path. Since the problems we are interested in are either trivially infeasible on disconnected graphs, or can be reduced to studying the components of connectivity one by one, we only consider connected input graphs in the sequel. 

A {\em vertex  cut} in a graph $G$ is a set $A\subset V(G)$ of vertices such that the graph $G-A=G[V(G)\setminus A]$ is disconnected. The {\em vertex connectivity} $c_v(G)$ of a graph $G$ is the order of a minimum cut in $G$, or $|V(G)|-1$ if $G$ is a complete graph. Since we will not consider edge connectivity, we will often omit the adjective when talking about the connectivity measure, we always have vertex connectivity in mind. 

Although we consider undirected graphs, when we talk about a path in a graph, the path itself is considered traversed in the direction from its starting vertex to the ending one. Formally, when we say that a path $P$ connects a vertex $x$ to a vertex $y$, then by $P^{-1}$ we denote the same path, but traversed from $y$ to $x$. This is important when creating a longer path by concatenating shorter ones.  

We will be using the following well-known corollary of a theorem of Menger.

\begin{proposition}\label{prop:menger}  \cite{denley2001generalization}
Let $G$ be an $s$-connected graph with $s \geq 1$. If $x \in V(G)$, $Y \subseteq V(G)$ and $x \not \in Y$, then there exist distinct vertices $y_1, \ldots, y_m \in Y$, where $m= \min \{s, |Y|\}$, and internally disjoint paths $P_1, \ldots, P_m$ such that for every $i \in 1,\ldots,m$,
\begin{itemize}
\item 
$P_i$ is a path starting in vertex $x$ and ending in vertex $y_i$, and 
\item 
$P_i \cap Y = \{y_i\}$.
\end{itemize}
\end{proposition}

We build upon the following results of Chv\'atal and Erd\H{o}s which we will later generalize to Hamiltonian linkages in Theorem~\ref{thm:Ham-linked}.

\begin{proposition} \cite{chvatal1972note} \label{prop:chvatal} 
Let $G$ be an $s$-connected graph. 
\begin{enumerate}
\item If $\alpha(G)<s+2$, then $G$ has a Hamiltonian path. 
\item If $\alpha(G)<s+1$, then $G$ has a Hamiltonian cycle.
\item If $\alpha(G)<s$, then $G$ is Hamiltonian connected (i.e., every pair of vertices is joined by a Hamiltonian path).
\end{enumerate}
\end{proposition}

\section{Hamiltonian linkage in graphs} \label{sec:linkage}

In order to prove the main result of the paper, i.e., to show that Hamiltonian path and Hamiltonian cycle are solvable in polynomial time for graphs of bounded independence number, we introduce and study a more general problem of {\em Hamiltonian linkages} in graphs. 

\subsection{Definition and existence}\label{subsec:existencelinkage}

The following notion of {\em $\ell$-linked} graphs is well known and frequently used in the study of graph minors. 

\begin{definition} 
A graph $G$ is \emph{$\ell$-linked} if it has at least $2\ell$ vertices and for every sequence of distinct vertices $s_1, \ldots, s_{\ell}, t_1, \ldots, t_{\ell}$ (commonly referred to as the {\em terminal vertices}), there exist $\ell$  vertex disjoint paths $P_1, \ldots, P_{\ell}$ such that the end-vertices of each $P_i$ are $s_i$ and $t_i$, $i=1,2,\ldots,\ell$. Such a collection of paths is called a \emph{linkage} of the vertices $s_1, \ldots, s_{\ell}, t_1, \ldots, t_{\ell}$.
\end{definition}

For $\ell=1$, a graph is 1-connected if and only if it is 1-linked. In general, $\ell$-connectedness does not guarantee $\ell$-linkedness, but sufficiently high vertex connectivity does: 

\begin{proposition}\cite{thomas2005improved} \label{prop:linked}
If a graph $G$ is vertex $10\ell$-connected, then $G$ is $\ell$-linked. 
\end{proposition}

A graph is called {\em Hamiltonian connected} if any two distinct vertices are connected by a Hamiltonian path. We introduce the following generalization of this notion to linkages:

\begin{definition}
A graph $G= (V,E)$ is \emph{Hamiltonian-$\ell$-linked} if it has at least $2\ell$ vertices and for every sequence of distinct vertices $s_1, \ldots, s_{\ell}, t_1, \ldots, t_{\ell}$, there exists an $\ell$-linkage with paths $P_1, \ldots, P_{\ell}$ such that $\bigcup_{i=1}^{\ell}V(P_i) = V$.
\end{definition}

As stated in Proposition~\ref{prop:chvatal}.3, Chv\'atal and Erd\H{o}s~\cite{chvatal1972note} proved that  a graph is Hamiltonian connected if its vertex-connectivity exceeds its independence number. Here we generalize this result to linkages.

\begin{theorem}\label{thm:Ham-linked}
For every $k>1$ and $\ell>0$, if a graph $G$ satisfies  $\alpha(G) < k$ and $c_v (G) \geq g(k,l)= \textrm{max} \{ k \ell, 10 \ell \}$, 
then $G$ is Hamiltonian-$\ell$-linked.
\end{theorem}

\begin{proof}
Let $G$ be a graph  such that $\alpha(G) < k$ and $c_v (G) \geq g(k,l)$. Since $g(k,\ell)\ge 10\ell$, $G$ has at least $10\ell+1>2\ell$ vertices. Let $s_1, \ldots, s_\ell, t_1, \ldots, t_\ell$ be a sequence of distinct vertices of $G$. By Proposition~\ref{prop:linked}, there exists a linkage between these vertices. Take a linkage ${\cal P} = \{P_1, \ldots, P_\ell\}$ such that the number of vertices contained in this linkage is maximum possible and set $P = \bigcup_{i=1}^{\ell} V(P_i).$

We claim that this linkage is Hamiltonian. Suppose for the contrary that there is a vertex $x$ which is not contained in any path $P_i$. Since $G$ is at least $k\ell$-connected,  by Proposition~\ref{prop:menger}, there exist at least $m= \min \{k\ell, |P|\}$ vertex disjoint (apart from $x$) paths $R_1, \ldots, R_{m}$ such that each $R_i$ is a path starting in vertex $x$ and ending in a vertex $p_i\in P$, and  $R_i \cap P = \{p_i\}$ for all $i \in 1,\ldots,m$. Thus each $R_i$, $i \in 1,\ldots,m$, shares exactly one vertex with exactly one path from $\cal P$.

If $k\ell \ge  |P|$, then $m=|P|$ and for every vertex in $P$ there exists a path $R_i$ ending in it. Then the linkage could be extended to include $x$ contradicting the maximality. Choose two consecutive vertices $u,v$ on $P_1$ and replace $P_1$ by $P'_1=s_1,\ldots,u,R_{\alpha}^{-1},x,R_{\beta},v,\dots,t_1$, where $\alpha$ and $\beta$ are indices such that $p_{\alpha}=u$ and $p_{\beta}=v$. Obviously, ${\cal P}'=\{P'_1,P_2,\ldots,P_{\ell}\}$ is a linkage for $s_1,\ldots,s_{\ell},t_1,\ldots,t_{\ell}$ and $P'=V(P'_1)\cup \bigcup_{i=2}^{\ell} V(P_i)$ contains more vertices than $P$.

Suppose that $ k\ell < |P|$. Then $m=k\ell$. By the pigeon-hole principle, there is a path $P_i$ 
such that at least $k$ paths $R_{i_1}, \ldots, R_{i_k}$ end in $P_i$. Without loss of generality assume that the indices $i_j,j=1,2,\ldots,k$ are ordered so that the end-vertices $p_{i_1}, p_{i_2},\ldots,p_{i_k}$ occur in this order on the path $P_i$, when traversed from $s_i$ to $t_i$. Denote by $y_j$ the successor of $p_{i_j}$ on $P_i$, $j=1,2,\ldots,k-1$ (note here that though $p_{i_k}$ may coincide with vertex $t_i$ and thus would not have any successor, for each of the other vertices $p_{i_1}, p_{i_2},\ldots,p_{i_{k-1}}$ the successor exists). See Figure~\ref{fig:cesty} for an illustration.
Consider the $k$ vertices $x,y_1,\ldots, y_{k-1}$. Since $\alpha(G)<k$, they cannot form an independent set.

\begin{figure}
\center
\includegraphics[scale=1]{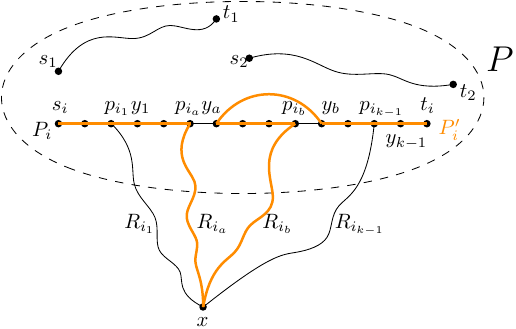}
\caption{An illustration of situation in the proof of Theorem \ref{thm:Ham-linked}.}
\label{fig:cesty}
\end{figure} 

If $x$ is adjacent to one of the $y$-vertices, say $xy_j\in E(G)$ for a $j\in\{1,\ldots,k-1\}$, we replace $P_i$ by the path $P'_i=s_i,\ldots,p_{i_j},R_{i_j}^{-1},x,y_j,\ldots,t_i$ and conclude that ${\cal P}'=\{P_1,P_2,\ldots,P_{i-1},P'_i,P_{i+1},\ldots,P_{\ell} \}$ is a linkage for $s_1,\ldots,s_{\ell},t_1,\ldots,t_{\ell}$ which traverses more vertices than $\cal P$, contradicting the assume maximality of $\cal P$.

If two of the $y$-vertices are adjacent, say $y_ay_b\in E(G)$ for some $a<b$, we replace $P_i$ by $P_i'=s_i,\ldots, 
p_{i_a},R_{i_a}^{-1},x,R_{i_b},p_{i_b},\ldots,y_a,y_b,\ldots,t_i$. Again as in the previous case, we conclude that ${\cal P}'=\{P_1,P_2,\ldots,P_{i-1},P'_i,P_{i+1},\ldots,P_{\ell} \}$ is a linkage for $s_1,\ldots,s_{\ell},t_1,\ldots,t_{\ell}$ which traverses more vertices than $\cal P$, contradicting the assumed maximality of $\cal P$. (A special case is when $b=a+1$ and $y_a=p_{i_b}$. Then $P_i'=s_i,\ldots, 
p_{i_a},R_{i_a}^{-1},x,R_{i_b},p_{i_b},y_b,\ldots,t_i$.) This concludes the proof of the claim that $\cal P$ is a Hamiltonian linkage.
\end{proof}

Note that the proof of Theorem~\ref{thm:Ham-linked} is constructive and yields a polynomial-time algorithm for constructing a Hamiltonian $\ell$-linkage. For more, cf. Subsection~\ref{subsec:construct} of the Conclusion section. 

Next we will attend to the computational complexity of Hamiltonian linkages in graphs. 
We will consider the following problems which differ in whether the terminal 
vertices are specified as part of the input (referred to as the {\sc Hamiltonian-$\ell$-Linkage} problem), or not, in which case we may ask if a Hamiltonian linkage exists for any selection of $\ell$ pairs of terminal 
vertices (the {\sc Hamiltonian-$\ell$-Linkedness} problem), or for at least one such collection of vertices (for historical reasons, this last problem is referred to as the {\sc $\ell$-Path-Cover} problem):

\computationproblem{{\sc Hamiltonian-$\ell$-Linkage}}{A graph $G$ and distinct vertices $s_1,\ldots,s_{\ell},t_1,\ldots,t_{\ell}\in V(G)$.}{Does there exist a Hamiltonian linkage for $s_1,\ldots,s_{\ell},t_1,\ldots,t_{\ell}$ in $G$?}

\computationproblem{{\sc Hamiltonian-$\ell$-Linkedness}}{A graph $G$.}{Is $G$ Hamiltonian-$\ell$-linked? }

\computationproblem{{\sc $\ell$-Path-Cover}}{A graph $G$.}{Do there exist distinct vertices   $s_1,\ldots,s_{\ell},t_1,\ldots,t_{\ell}\in V(G)$ and a Hamiltonian linkage for them in $G$?}

For $\ell=1$, {\sc Hamiltonian-$1$-Linkage} is the problem whether the input graph contains a Hamiltonian path connecting two vertices given as part of the input, {\sc Hamiltonian-$1$-Linkedness} is the problem asking if the input graph is Hamiltonian connected, and the {\sc $1$-Path-Cover} is asking if the input graph contains a Hamiltonian path. The last mentioned problem is well known to be NP-complete on general graphs (and on many narrower graph classes, cf. the Introduction section). The NP-completeness of deciding if a specified pair of vertices can be connected by a Hamiltonian path is shown NP-complete already in the seminal book of Garey and Johnson \cite{Garey:2000}. The NP-completeness of Hamiltonian connectedness is attributed to \cite{Garey:2000} and \cite{dean1993computational} in~\cite{kuvzel2012thomassen}, we will give an explicit argument in the next subsection. We will show that all three above defined problems are solvable in polynomial time for every fixed $\ell$ on graphs of bounded independence number. But first we will show that on general graphs, all three of them are NP-complete for every fixed $\ell$.

\subsection{Hamiltonian linkage in general graphs}

We use the well-known fact that deciding the existence of a Hamiltonian path in a connected graph is NP-complete, and we will gradually reduce this problem to all variants of the Hamiltonian linkage problem.

\begin{theorem}
For every fixed $\ell>0$, {\sc Hamiltonian-$\ell$-Linkage} is NP-complete on connected graphs.
\end{theorem} 

\begin{proof}
Let $G$ be a connected graph, subject to the question if it contains a Hamiltonian path. Construct $G'$ by adding $2\ell$ new extra vertices $s_1, \ldots, s_{\ell}, t_1,\ldots, t_{\ell}$ and edges $s_1s_i, s_it_i$ and $t_it_1$ for all $ i=2,\ldots,\ell$ and $s_1v,t_1v$ for all  $v\in E(G)$. In any linkage of $s_1,\ldots,t_{\ell}$ (not necessarily a Hamiltonian one), the vertices $s_i,t_i$, $i=2,3,\ldots,\ell$ are linked via the edge $s_it_i$, since any other path in $G'$ connecting these vertices would pass through $s_1$ and $t_1$ and these two vertices could not be linked. Thus $s_1$ and $t_1$ must be linked via a path in $G$, which must visit all vertices of $G$ if the linkage is Hamiltonian. It follows that $s_1,\ldots,t_{\ell}$ are Hamiltonian linked in $G'$ if and only if $G$ contains a Hamiltonian path.   
\end{proof}

\begin{theorem}\label{thm:NPhard-PathCover}
For every fixed $\ell>0$, {\sc $\ell$-Path-Cover} is NP-complete on connected graphs.
\end{theorem}

\begin{proof}
Let $G$ be a connected graph, subject to the question if it contains a Hamiltonian path. Construct $G'$ by adding $\ell+3$ new extra vertices $a,b,c,d_1,\ldots,d_{\ell}$ and edges $au, u\in V(G)$, $ab,bc$ and $cd_i, i=1,2,\ldots,\ell$. In any path cover, at most two of the vertices $d_i, i=1,2,\ldots,\ell$ belong to the same path, since such a path passes through vertex $c$. If this is the case, then a path cover by $\ell$ can be achieved only by making all vertices of $V(G)\cup\{a,b\}$ be covered by a single path, which would have to start in $b$, and this is possible only if (and if) $G$ itself contains a Hamiltonian path. If each vertex $d_i, i=1,2,\ldots,\ell$ belongs to a different path, we have already $\ell$ of them, and so one of them must pass through $c, b$ and $a$ and continue as a Hamiltonian path in $G$. Thus $G'$ can be covered by at most $\ell$ paths if and only if $G$ contains a Hamiltonian path, and the path cover number of $G'$ is exactly $\ell$ in such a case.
\end{proof}

\begin{theorem}
For every fixed $\ell>0$, {\sc Hamiltonian-$\ell$-Linkedness} is NP-complete on connected graphs.
\end{theorem}

\begin{proof}
Let us first briefly argue that for a fixed $\ell$, the {\sc Hamiltonian-$\ell$-Linkedness} problem is in NP, though it asks about the existence of Hamiltonian linkage {\em for any} collection of $\ell$ pairs of vertices of the input graph. For {\sc Hamiltonian-$\ell$-Linkage}, the guess-and-verify certificate to prove NP-membership is simply a collection of $\ell$ paths connecting the terminal 
vertices, 
including checking that every vertex of the input graph belongs to exactly one of these paths. For   {\sc Hamiltonian-$\ell$-Linkedness}, a guess-and-verify certificate is a list of $O(n^{2\ell})$ such collections of paths, where $n$ is the order of the input graph. For a fixed $\ell$, this is a polynomial-time verifiable certificate.

Now we prove the NP-hardness by a reduction from {\sc $\ell$-Path-Cover} which is NP-complete according to Theorem~\ref{thm:NPhard-PathCover}. Given an input graph $G$ with $n$ vertices subject to the question if $G$ can be covered by $\ell$ disjoint paths, we may assume that $n>3\ell$, since we can solve the problem by brute force (which would still take only $O(1)$ time) otherwise. We construct a graph $G'$ by adding $2\ell$ new extra vertices $x_1,\ldots,x_{2\ell}$ to $G$, and making each of them adjacent to all vertices of $G$. (Formally, $V(G')=V(G)\cup  X$, where $X=\{x_1,\ldots,x_{2\ell}\}$, and $E(G')=E(G)\cup\{xu: x\in X, u\in V(G)\}$.) We claim that $G'$ is Hamiltonian-$\ell$-linked if and only if $G$ can be covered by $\ell$ disjoint paths.

For the ``only if'' direction, suppose that $G'$ is Hamiltonian-$\ell$-linked. Then the vertices $s_i=x_i, t_i=x_{\ell+i}, i=1,2,\ldots,\ell$ are Hamiltonian linked in $G'$, and in such a linkage, each path $P_i=s_i\ldots t_i$ induces a path $A_i=P_i\cap V(G)$ in $G$ and these paths form a path cover of $G$.

The proof of the ``if'' part is slightly more involved. Suppose $G$ can be covered by $\ell$ disjoint paths $A_1,\ldots,A_{\ell}$. We will show that any collection of vertices $s_1,\ldots,t_{\ell}$ in $G'$ are Hamiltonian linked in $G'$. Let $s_1,\ldots,t_{\ell}\in V(G')$ be such a collection, for the sake of brevity, we will call these vertices {\em marked}. Let $k$ of the pairs $s_i,t_i$ have both vertices in $X$ and let $h$ of them have only one vertex in $X$. The remaining $\ell-k-h$ pairs have both vertices in $V(G)$. That means that $X$ has $2\ell-2k-h$ vertices which are not marked, and we will call these vertices {\em free}.         

For each path $A_i$, choose its orientation from left to right and partition its vertices into {\em segments} in the following way: For every marked vertex $u\in V(A_i)$, the segment $S(u)$ contains $u$ and the maximal subpath of $A_i$ which contains vertices to the left of $u$ and for which $u$ is the only marked vertex of this segment. If the right end-vertex of the path $A_i$ is not marked, we create one more segment containing the vertices lying on $A_i$ to the right from its right-most marked vertex, and we call this segment {\em free}.

Having done this with all paths $A_i,i=1,2,\ldots,\ell$, we get the vertices of $G$ partitioned into $2\ell-2k-h$ marked segments and at most $\ell$ free ones. Now we modify the segments so that we have exactly $\ell$ of the free ones. Since $n>3\ell$, we have at least $\ell$ vertices in $G$ which are not marked. Some of them will be taken away from their segments and pronounced single-vertex free segments so that the all the segments are paths in $G$, they form a partition of $V(G)$, there are exactly $\ell$ free segments and every marked segment is either a single marked vertex or a path containing exactly one marked vertex, which is then an endpoint of this segment. To achieve this, we simply remove the necessary number of unmarked vertices from their former segments from left to right.

Finally we show how to build a Hamiltonian linkage for the marked vertices. For a pair of marked vertices $s_i,t_i\in X$, choose a free segment $S$ in $G$ and take the path $P_i=s_i,S,t_i$ (and delete $S$ from the pool of free segments). For a pair of marked vertices $s_i,t_i$ such that exactly one of them, say $s_i$, is in $X$, tak the path $P_i=s_i,S(t_i)$. For a pair of marked vertices $s_i,t_i\in V(G)$, choose a free vertex $x\in X$ and take the path $P_i=S(s_i)^{-1},x,S(t_i)$. Having done this for all $i=1,\ldots,\ell$, we have created disjoint paths $P_i,i=1,2,\ldots,\ell$ in $G'$, and these paths contain all vertices of $G'$ except for $2\ell-2k-h-(\ell-k-h)=\ell-k$ free vertices of $X$ (say, $y_1,\ldots,y_{\ell-k}$) and $\ell-k$ free segments in $G$ (which were not used so far, say $S_1,\ldots,S_{\ell-k}$). The path $P_1$ contains two consecutive vertices $u\in X, v\not\in X$ (or  $u\not\in X, v\in X$, which is analogous) and we modify $P_1$ by inserting the sequence $S_1,y_1,S_2,y_2,\ldots,S_{\ell-k},y_{\ell-k}$ between $u$ and $v$. Now the paths $P_1,\ldots,P_{\ell}$ form a path cover of $G'$. This concludes the proof.    
\end{proof}

\subsection{Hamiltonian linkage in graphs of bounded independence number}

Now we prove the main result of this section. 

\begin{theorem}\label{thm:Ham-linkage-poly}
For every $k>0, \ell>0$, there exists an integer $f(k,\ell)<\infty$ such that {\sc Hamiltonian-$\ell$-Linkage}
can be solved in time $O(n^{f(k,\ell)})$ for $n$-vertex input graphs of independence number smaller than $k$.  
\end{theorem}

We will first describe the algorithm informally. Suppose we are given a $kK_1$-free graph $G$ and distinct vertices $s_1,\ldots,s_{\ell},t_1,\ldots,t_{\ell}$. We first compute the vertex connectivity $c_v(G)$ of $G$. If $c_v(G) \geq g(k,l)$, we answer ``yes" and quit, since by Theorem~\ref{thm:Ham-linked}, $G$ is Hamiltonian-$\ell$-linked, and hence a Hamiltonian linkage exists also for the vertices $s_1,\ldots,s_{\ell},t_1,\ldots,t_{\ell}$. 

Otherwise, we find and fix a vertex cut $A$ such that $|A| < g(k,l)$. Let $Q_1,\ldots,Q_s$ be the connected components of $G-A$. Since $G$ does not contain an independent of size $k$, we observe that $s < k$. On the other hand, $s\ge 2$ because $A$ is a cut. Consider a component $Q_j$. If $I_j$ is an independent set in $G[Q_j]$, adding one vertex per every other component of $G-A$ to $I_j$ results in an independent set in $G$ of size $|I_j|+s-1$. It follows that $\alpha(G[Q_j])<k-(s-1)\le k-1$ for every $j=1,2,\ldots,s$.

For any path $P$ in $G$, define its {\em scenario} $sc(P)$ as follows. Traverse the path from its starting point, say $\widetilde{s}$, to its end-point, say $\widetilde{t}$, and write down a reduced sequence of vertices consisting of the vertices of the cut $A$, their immediate neighbours  in the path $P$ and the first and last vertices $\widetilde{s}$ and $\widetilde{t}$.  For example, if $P=\widetilde{s},a,b,c,d,e,f,g,h,o,p,q,r,y,z,\widetilde{t}$ and $\widetilde{s},a,b,c\in Q_1$, $d,e,f\in A$, $g,h,o\in Q_2$, $p\in A$, $q\in Q_1$, $r\in A$, $y,z,\widetilde{t}\in Q_1$, then the scenario of $P$ is $sc(P)=\widetilde{s},c,d,e,f,g,o,p,q,r,y,\widetilde{t}$. The {\em blind scenario} $bsc(P)$ is obtained from the scenario by replacing vertices from components $Q_j$ (except for the first and last vertex) by the names of the components, and leaving only one occurrence of $Q_j$ from every two consecutive ones. In case of the example above, the blind scenario is $bsc(P)=\widetilde{s},Q_1,d,e,f,Q_2,p,Q_1,r,Q_1,\widetilde{t}$. See Figure~\ref{fig:scenario}.  

\begin{figure}
\center
\includegraphics[scale=1]{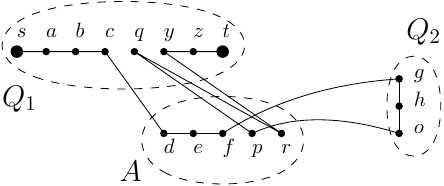}
\caption{An example of a scenario in the proof of Theorem \ref{thm:Ham-linkage-poly}.}
\label{fig:scenario}
\end{figure}

We will guess (i.e., try all possible) scenarios of the paths $P_1,P_2,\ldots,P_{\ell}$ for $s_1,\ldots,s_{\ell},$ $t_1,\ldots,t_{\ell}$ and for each such collection check if this collection of scenarios can be realized by a Hamiltonian linkage. There are some obvious necessary conditions to be satisfied - every vertex of $A$ must belong to exactly one $sc(P_i)$ (and appear only once in it), every vertex of $V(G)\setminus A$ belongs to at most one $sc(P_i)$ (and appears only once in it), and any two consecutive vertices in every scenario are adjacent in $G$.
We call a collection of scenarios {\em plausible} if these necessary conditions are fulfilled.

For a plausible collection of scenarios $sc(P_1),sc(P_2),\ldots,sc(P_{\ell})$, we process one by one the components of $G-A$. For each component $Q_j$ we note the pairs of vertices which occur consecutively in some scenario and we mark these two vertices to be linked in a recursively called {\sc Hamiltonian-$\ell'$-Linkage} subroutine on $G[Q_j]$. (There is one technicality here. If a pair of terminal 
vertices coincide, we simply delete this vertex from $Q_j$ for the subroutine.) If a Hamiltonian linkage  ${\cal P}=\{P_1,P_2,\ldots,P_{\ell}\}$ realizes a collection of scenarios, the restrictions of the paths from the linkage to a component $Q_j$ provide a Hamiltonian linkage for the recursive subroutine. And vice versa, if the subroutines for all components return the answer "yes", the union of their linkages provides a Hamiltonian linkage for $s_1,\ldots,s_{\ell},t_1,\ldots,t_{\ell}$ in $G$. The algorithm is formally described in pseudocode as   Algorithm~\ref{alg:algorithm} with an auxiliary subroutine Algorithm~\ref{alg:subroutine}. For its sake, we define the scenarios formally.

\begin{definition}
Let $G$ be a graph, let $A\subset V(G)$ be a vertex cut and let $Q_1,\ldots,Q_s$ be the connected components of $G-A$. Let $s_1,\ldots,s_{\ell}, t_1,\ldots,t_{\ell}$ be distinct vertices of $G$. A {\em scenario} for $s_i,t_i$ is a sequence of vertices of $G$ starting with $s_i$ and ending with $t_i$ such that no two consecutive elements belong to different components $Q_j,Q_{j'}$, $j\neq j'$, and no three consecutive elements belong to the same component $Q_j$ (in other words, between any two closest occurrences of vertices from $A$, there are at most 2 vertices from the components $Q_j$, and if they are 2, then they are from the same component). A {\em collection of scenarios} for $s_1,\ldots,s_{\ell}, t_1,\ldots,t_{\ell}$ is a set $\cal S$ of $\ell$ scenarios $S_1,\ldots,S_{\ell}$ such that for each $i$, $S_i$ is a scenario for $s_i,t_i$. A scenario collection ${\cal S}=\{S_1,\ldots,S_{\ell}\}$ is {\em plausible} if every vertex of $A$ belongs to exactly one scenario $S_i$ and appears only once in it, every vertex of $G-A$ belongs to at most one scenario $S_i$ and then appears only once in it, and in every scenario $S_i$, any two consecutive elements that do not both belong to the same component $Q_j$ are adjacent in $G$.  
\end{definition}

\begin{algorithm}
\caption{$HamLinkage(G,k,\ell, s_1,\ldots,s_{\ell},t_1,\ldots,t_{\ell})$.}\label{alg:algorithm}
\begin{algorithmic}[1]
\Require A graph $G$, a non-negative integer $\ell$, distinct vertices $s_1,\ldots,s_{\ell},t_1,\ldots,t_{\ell}\in V(G)$.
\Ensure \texttt{true} if $G$ contains a Hamiltonian linkage for $s_1,\ldots,s_{\ell},t_1,\ldots,t_{\ell}$, \texttt{false} otherwise.
\If{$\ell=0$ and $V(G)\neq\emptyset$}
    \State \Return \texttt{false}
\EndIf
\If{$\ell=0$ and $V(G) = \emptyset$}
    \State \Return \texttt{true}
\EndIf
\State \textbf{compute} the vertex connectivity $c_v(G)$ of $G$ and a minimum vertex cut $A\subseteq V(G)$
\If{$c_v(G)\ge g(k,\ell)$} 
		\State \Return \texttt{true}
	\Else
		\State \textbf{denote}  by $Q_j, j=1,2,\ldots,s$ the connected components of $G-A$;
		\State $HamLinked \gets$ \texttt{false}
		\ForAll {plausible scenario collections $\cal S$ $=\{S_1,S_2,\ldots,S_{\ell}\}$}
			\State $GoodScenario \gets$ \texttt{true}
			\For{$j = 1, \dots, s$}
				\State $Reduce(Q_j,{\cal S},Q',\ell', s'_1,\ldots,t'_{\ell'})$
				\If{$HamLinkage(G[Q'],k-1,\ell',s'_1,\ldots,t'_{\ell'})=$ \texttt{false}}
					\State $GoodScenario \gets$ \texttt{false}
				\EndIf
			\EndFor
			\If{$GoodScenario$}
					\State $HamLinked \gets$ \texttt{true}
			\EndIf		
		\EndFor 
		\If {$HamLinked$}
			\State \Return \texttt{true}
		\Else 
			\State \Return \texttt{false}
		\EndIf
\EndIf 
\end{algorithmic}
\end{algorithm}

\begin{algorithm}
\caption{$Reduce(Q,{\cal S},Q',\ell',s'_1,\ldots,t'_{\ell'})$.}\label{alg:subroutine}
\begin{algorithmic}[1]
\Require A connected component $Q$ of $G-A$, a plausible scenario collection ${\cal S}=\{S_1,\ldots,S_{\ell}\}$ for $s_1,\ldots,t_{\ell}$.
\Ensure  A subset $Q'\subseteq Q$, an integer $\ell' \geq 0$, distinct vertices $s'_1,\ldots,s'_{\ell'},t'_1,\ldots,t'_{\ell'}\in Q'$.
\State $Q' \gets Q$; $\ell' \gets 0$
\For{ $i=1, \ldots, \ell$}
	\ForAll {$u\in S_i\cap Q$ such that no neighbor of $u$ on $S_i$ belongs to $Q$}
		\State $Q'\gets Q'\setminus\{u\}$
	\EndFor	
	\ForAll { $u,v\in Q$ that are distinct and appear consecutively on $S_i$}
		\State $\ell'\gets \ell'+1$; $s'_{\ell'}\gets u$; $t'_{\ell'}\gets v$
	\EndFor	
\EndFor
\State \Return $(Q',\ell',s'_1,\ldots,t'_{\ell'})$

\end{algorithmic}
\end{algorithm}

\begin{lemma}
Algorithm~\ref{alg:algorithm} correctly answers if a $kK_1$-free graph contains a Hamiltonian linkage for the input vertices $s_1,\ldots,t_{\ell}$. 
\end{lemma}

\begin{proof}
The correctness has been argued upon in the outline of the algorithm. We will justify the $Reduce$ subroutine. The correspondence between a path $P$ and its scenario $S$ is such that if two consecutive vertices of $S$ belong to the same component $Q_j$, they correspond to a path connecting these vertices in $G[Q_j]$ as a subpath of $P$. Therefore these two vertices are added as a pair to the terminal 
vertices for the recursive call. If there is only one vertex of $Q_j$ between two vertices from $A$ on $S$, say $a_1,q,a_2$ are consecutive vertices with $a_1,a_2\in A, q\in Q_j$, then $a_1,q,a_2$ is a subpath of $P$ and vertex $q$ cannot be used by any other path of the linkage. Therefore we delete this vertex from $Q_j$ for the recursive call. It is then clear that a collection of scenarios can be realized by a linkage of the terminal 
vertices if and only if it is plausible and the recursive calls on the reduced components all return value {\texttt{true}}. And that the vertices can be linked if and only if there exists at least one plausible scenario collection that can be realized by a linkage. 
\end{proof}

\begin{proof} {\bf (of Theorem~\ref{thm:Ham-linkage-poly})}
We will estimate the running time of Algorithm~\ref{alg:algorithm} and show by induction on $k$ the claim of the theorem. Suppose $G$ has $n$ vertices. For $k=2$, the input graph is complete and as such it does contain a Hamiltonian linkage. Hence $f(2,\ell)=2$, even if we do not trust the promise and check that $\alpha(G)<2$.

Now suppose $k>2$. The vertex connectivity of $G$ can be computed in time $O(n^4)$ \cite{henzinger2000computing}. 
If the connectivity $c_v(G)$ is less than $g(k,\ell)$, a small vertex cut $A$ is identified as a byproduct of the connectivity determining algorithm.

A rough estimate of the number of plausible collections of scenarios can go as follows. For a single pair of terminal 
vertices $s_i,t_i$, the number of blind scenarios is at most $|A|!\cdot s^{|A|-1}   \le (kg(k,\ell))^{g(k,\ell)}$. Thus the number of blind scenario collections is at most $(kg(k,\ell))^{\ell g(k,\ell)}$. If a blind scenario uses $h$ vertices from the cut $A$, we have at most $2h$ positions as placeholders for vertices from the components. As every vertex of $A$ appears in exactly one scenario, altogether we have at most $2|A|$ placeholders for vertices from the components, and each placeholder can be filled in at most $n$ ways. So all together we have at most $n^{2|A|}\cdot  (kg(k,\ell))^{\ell g(k,\ell)}\le n^{2g(k,\ell)}\cdot (kg(k,\ell))^{\ell g(k,\ell)}$ collections of scenarios. 

For each of them, we call $s\le k-1$ subroutines of Algorithm~\ref{alg:subroutine} for the components, and then call Algorithm~\ref{alg:algorithm} recursively for the reduced components. On the scenarios, we have together at most $|A|+\ell$ placeholder gaps between occurrences of vertices from $A$, and each such gap is either a single vertex that gets deleted from its component, or a pair of vertices which are added to the terminal 
vertices of this component. Hence reducing the components takes $O(g(k,\ell)+\ell)$ time. The recursive call takes time $O(n^{f(k-1,\ell')})$ by induction hypothesis, because $\alpha(G[Q'_i])\le k-2$ for every $i$. The $\ell'$ for the recursive call may be much bigger than $\ell$, but still can be bounded. Even if all placeholders were for the same component (which is not possible), we would have at most $|A|+\ell$ pairs of terminal 
vertices in this component, and thus $\ell'<g(k,\ell)+\ell$. The overall running time is thus $$O(n^4)+O(n^{2g(k,\ell)}\cdot (kg(k,\ell))^{\ell g(k,\ell)}\cdot (g(k,\ell)+\ell+sn^{f(k-1,g(k,\ell)+\ell)}))=O(n^{2g(k,\ell)+f(k-1,g(k,\ell)+\ell)}),$$
where the multiplicative constant in the $O$ notation depends on $k$ and $\ell$. Nevertheless, it follows by induction on $k$ that
$f(k,\ell)\le  2g(k,\ell)+f(k-1,g(k,\ell)+\ell)<\infty$.     
\end{proof}

\begin{theorem}\label{thm:Ham-Linked-poly}
For every $k>0,\ell>0$, the {\sc Hamiltonian-$\ell$-Linkedness} and {\sc $\ell$-Path-Cover } problems
can be solved in polynomial time  for input graphs of independence number smaller than $k$.  
\end{theorem}

\begin{proof}
Given a graph with $n\ge 2\ell$ vertices, we check all collections of $\ell$ pairs of terminal 
vertices, and for each such collection, we decide if they can be Hamiltonian linked or not by the algorithm of Theorem~\ref{thm:Ham-linkage-poly}. The graph is Hamiltonian-$\ell$-linked if and only if the answer is ``yes" for all the collections, and the graph allows a path cover by $\ell$ paths if and only if the answer is ``yes" for at least one of the collections. There are at most $n^{2\ell}$ such collections, and thus {\sc Hamiltonian-$\ell$-Linkedness} and {\sc $\ell$-Path-Cover} can be answered in time $O(n^{2\ell+f(k,\ell)})$, where $f(k,\ell)$ is the function from Theorem~\ref{thm:Ham-linkage-poly}.
\end{proof}

\begin{corollary}\label{cor:vsechnopolynomialni}
For every $k>0$, deciding the existence of a Hamiltonian cycle, a Hamiltonian path, or Hamiltonian connectedness of a $kK_1$-free graph can be performed in polynomial time. 
\end{corollary}

\begin{proof}
Set $\ell=1$ in Theorem~\ref{thm:Ham-linkage-poly} and check all pairs of vertices of the input graph if this pair is Hamiltonian linked. This can be done in time $O(n^{2+f(k,1)})$. The graph contains a Hamiltonian path if the answer is ``yes" for at least one pair of vertices, it contains a Hamiltonian cycle if the answer is ``yes" for at least one pair of adjacent vertices, and it is Hamiltonian connected if the answer is ``yes" for all the pairs. 
\end{proof}

\begin{corollary}\label{cor:pathcover}
For every $k>0$, the path cover number of a $kK_1$-free graph can be determined in polynomial time.
\end{corollary}

\begin{proof}
It is obvious that $\mathrm{pc}(G)\le \alpha(G)$. (Consider an optimal collection of $\mathrm{pc}(G)$ paths that cover all vertices and take one end-vertex from each of the paths. These vertices must form an independent set, since if any two of them were adjacent, these two paths could be replaced by their concatenation to form a cover by $\mathrm{pc}(G)-1$ paths. Thus $\alpha(G)\ge \mathrm{pc}(G)$.) If we assume that $G$ is $kK_1$-free, it suffices to call {\sc $\ell$-Path-Cover} on $G$ for $\ell=1,2,\ldots,k-1$ and return the smallest value of $\ell$ for which we get an affirmative answer. Hence the path cover number can be determined in time $O(k\cdot n^{f(k,k-1)})=  O(n^{f(k,k-1)}).$ 
\end{proof}

\section{Connection to $L(2,1)$-labelling}\label{sec:L21}

The so called Frequency Assignment Problem is  an area of problems motivated by assigning frequencies in mobile networks. One intensively studied variant is the so called distance constrained labelling, and in particular the $L(2,1)$-labelling. 

\begin{definition}
An \emph{$L(2,1)$-labelling} of a graph $G$ is a function $f$ from $V(G)$ to non-negative integers such that 
\begin{itemize}
\item $|f(u)-f(v)| \geq 2$ if $u$ and $v$ are adjacent, and 
\item $|f(u)-f(v)| \geq 1$ if the shortest path from $u$ to $v$ is of length 2. 
\end{itemize}
\end{definition}

A \emph{$k$-$L(2,1)$-labelling} is an $L(2,1)$-labelling such that there is no label greater than $k$. The \emph{$L(2,1)$-labelling number} of $G$, denoted by $\lambda(G)$, is the smallest possible $k$ such that $G$ has a $k$-$L(2,1)$-labelling.  

This notion was formally introduced by Griggs and Yeh in~\cite{griggs1992labelling}, where the first results and open questions were presented. The problem is rather interesting from the computational complexity point of view. Griggs and Yeh conjectured that it is NP-hard to determine $\lambda(G)$ even if $G$ is a tree. This is not the case, as shown by Chang and Kuo in~\cite{cha62} where a polynomial time algorithm for determining $\lambda(G)$ of trees is presented. Rather surprisingly, this algorithm does not generalize to graphs of bounded tree-width and $L(2,1)$-labelling belongs to a handful of problems polynomially solvable for trees but NP-hard already on graphs of tree-width 2, as shown in~\cite{FGK05}. An interesting connection to Hamiltonian paths and path covers of graphs was proved in~\cite{ge94}.    
Here and in the rest of the section, $\overline{G}$ denotes the complement of a graph $G$.

\begin{proposition}\cite{ge94} \label{prop:1} Let $G$ be a graph with $n$ vertices, then 
\begin{itemize}
\item $\lambda(G) \leq n-1$ if and only if $\mathrm{pc}(\overline{G}) = 1$,
\item $\lambda(G) = n+r-2$ if and only if $\mathrm{pc}(\overline{G}) = r$, where $r$ is an integer, $r \geq 2$.
\end{itemize}
\end{proposition}

A closely related and also studied problem requires that no label is used repeatedly.

\begin{definition}
An $L'(2,1)$-labelling of a graph $G$ is an $L(2,1)$-labelling where the function $f$ is injective.
\end{definition}

Similarly, a \emph{$k$-$L'(2,1)$-labelling} is an $L'(2,1)$-labelling such that there is no label greater than $k$. The \emph{$L'(2,1)$-labelling number} of $G$, denoted by $\lambda'(G)$, is the smallest possible $k$ such that $G$ has a $k$-$L'(2,1)$-labelling. It is clear that $\lambda'(G)\ge n-1$ for any graph with $n$ vertices. The following theorem is a consequence of the proof of Proposition~\ref{prop:1} as presented in~\cite{ge94}.

\begin{proposition} \label{prop:2} Let $G$ be a graph with $n$ vertices, then 
\begin{itemize}
\item $\lambda'(G) = n-1$ if and only if $\mathrm{pc}(\overline{G}) = 1$,
\item $\lambda'(G) = n+r-2$ if and only if $\mathrm{pc}(\overline{G}) = r$, where $r$ is an integer, $r \geq 2$.
\end{itemize}
\end{proposition}

The connection to Hamiltonicity and path covers is a strong tool when proving NP-hardness of $L(2,1)$-labelling restricted to special graph classes.
Let ${\cal G }$ be some graph class. 
If the problem \textsc{HamiltonianPath} is NP-complete on ${\cal G }$, then so is \textsc{PathCover} on ${\cal G }$ and hence, \textsc{$L(2,1)$-labelling} on ${\cal \overline  G }$ is also NP-complete. This line of reasoning was used in~\cite{BKTL00} to show that {\sc $L(2,1)$-Labelling} is NP-hard for planar graphs, bipartite graphs or chordal graphs. In Theorem~\ref{thm:L21} we will present a complete characterization of the complexity of this problem for $H$-free graphs, and an almost complete characterization for the  {\sc $L'(2,1)$-Labelling} problem in Theorem~\ref{thm:Lprime21} (we fully determine the complexity for triangle-free $H$). 
The following observation is straightforward, but we find it useful to be expressed formally.

\begin{observation}
Let $\textsc{P}$ be a computational problem. 
\begin{itemize}
\item
If $\textsc{P}$ restricted to the class of $H$-free graphs is polynomial time solvable, then for each induced subgraph $H'$ of $ H$, $\textsc{P}$ restricted to the class of $H'$-free graphs is also polynomial time solvable.

\item If $\textsc{P}$ restricted to the class of $H$-free graphs is NP-hard, then for each $H'$ which contains $H$ as an induced subgraph,  $\textsc{P}$ restricted to the class of $H'$-free graphs is also NP-complete.
\end{itemize}
\end{observation}

\begin{theorem}\label{thm:L21}
Let $H$ be a graph. If $H$ is an induced subgraph of $ P_4$ (the path of length 3), then \textsc{$L(2,1)$-labelling} restricted to the class of $H$-free graphs is polynomial-time solvable, and  it is NP-hard otherwise.
\end{theorem}

\begin{proof}
\textsc{$L(2,1)$-labelling} is polynomial-time solvable for the class of $P_4$-free graphs (called also cographs) \cite{cha62}.

Suppose that $H$ is not an induced subgraph of $P_4$. Bodlaender at al. \cite{bo04} proved that \textsc{$L(2,1)$-labelling} is NP-complete for bipartite graphs, chordal graphs and split graphs. 
Thus if $H$ contains a cycle or $2K_2$ as an induced subgraph, then \textsc{$L(2,1)$-labelling} restricted to the class of $H$-free graphs is NP-complete, since the class of $H$-free graphs is either a superclass of bipartite graphs (if $H$ contains a triangle), or of chordal graphs (if $H$ contains an induced cycle of length greater than 3), or of split graphs (if $H$ contains $2K_2$). 
Finally, let $H$ be a $2K_2$-free acyclic graph which is not an induced subgraph of $P_4$. Then $H$ necessarily contains $3K_1$ as an induced subgraph, and \textsc{$L(2,1)$-labelling} for $H$-free graphs is NP-complete because \textsc{HamiltonianPath} is NP-complete for the class of triangle-free graphs \cite{mu17}. 
\end{proof}

\begin{theorem}\label{thm:Lprime21}
Let $H$ be a triangle-free graph. If $H$ is an induced subgraph of $P_4$, then \textsc{$L'(2,1)$-labelling} restricted to the class of $H$-free graphs is polynomial-time solvable, and it is NP-hard otherwise.
\end{theorem}

\begin{proof}
\textsc{$L'(2,1)$-labelling} is polynomial-time solvable for the class of $P_4$-free graphs \cite{cha62}.

Suppose that $H$ is not an induced subgraph of $P_4$ and $K_3$ is not an (induced) subgraph of $H$. If $\alpha(H)\ge 3$, i.e., $H$ contains an induced copy of $3K_1$,  \textsc{$L'(2,1)$-labelling} is NP-complete for $H$-free graphs because \textsc{HamiltonianPath} is NP-complete for the class of triangle-free graphs \cite{mu17}.

If $\alpha(H)\le 2$, $H$ contains a cycle or $2K_2$ is an induced subgraph.
\textsc{HamiltonianPath} is NP-complete for the class of split graphs by \cite{mu96}. Split graphs are self-complementary and they form a subclass of  the class of co-chordal graphs. Hence, if $H$ contains an induced cycle of length 4 or 5, or an induced copy of $2K_2$, then \textsc{$L'(2,1)$-labelling} restricted to the class of $H$-free graphs is NP-complete. 
\end{proof}

If $H = K_3$, then the complement class of $H$-free graphs is the class of  $3K_1$-free graphs. It follows from \cite{br00} that each $3K_1$-free graph $G$ has a Hamiltonian path. Thus, $\lambda'(G)$ is equal to the number of vertices minus one. The general case of graphs $H$ containing triangles seems complicated. As the first step towards this goal, we observe a corollary of our result on path covers.

\begin{corollary}
For every $k$, \textsc{$L'(2,1)$-labelling} is polynomial time solvable in the class of $K_k$-free graphs.
\end{corollary}

\begin{proof}
For every $k$, {\sc PathCover} is solvable in polynomial time on $kK_1$-free graphs by Corollary~\ref{cor:pathcover}. 
\end{proof}

This leaves the following open problem.

\begin{problem}
Let $H$ be a graph which is not complete, but contains a triangle. Determine the complexity of \textsc{$L'(2,1)$-labelling} restricted to the class of $H$-free graphs. 
\end{problem}

\section{Conclusion}\label{sec:Conclusion}

\subsection{Maximum linkages, longest paths}

An optimization version of the Hamiltonian path problem is looking for the maximum possible length of a path in the input graph, or the maximum length of a path connecting two specified vertices. A straightforward generalization to Hamiltonian linkages is the {\sc Maximum-$\ell$-Linkage} problem which asks for the maximum number of vertices that can be covered by a linkage of given $\ell$ pairs of vertices. Our Algorithm~\ref{alg:algorithm} can be modified to solve this optimization version in the following way.

Call the {\em defect} of a graph $G$ with respect to vertices $s_1,\ldots,s_{\ell},t_1,\ldots,t_{\ell}$, denoted by $\mbox{def}(G,\ell,s_1,\ldots,s_{\ell},t_1,\ldots,t_{\ell})$, the minimum number of vertices that remain uncovered by a linkage of the vertices $s_1,\ldots,s_{\ell},t_1,\ldots,t_{\ell}$ in $G$. Set  $\mbox{def}(G,\ell,s_1,\ldots,s_{\ell},t_1,\ldots,t_{\ell})=\infty$ if $G$ does not contain any linkage of  the vertices $s_1,\ldots,s_{\ell},t_1,\ldots,t_{\ell}$. Then the algorithm $\mbox{MinDefect}(G,k,\ell,s_1,\ldots,s_{\ell},t_1,\ldots,t_{\ell})$ which will return a value from $\mathbb{Z}^+\cup\{0,\infty\}$
will run similarly as Algorithm~\ref{alg:algorithm}, with the following adjustments:

\begin{itemize}
\item The base case of $\ell=0$ is $\mbox{MinDefect}(G,k,0)\leftarrow |V(G)|$.
\item If $c_v(G)\ge g(k,\ell)$, we return $\mbox{MinDefect}(G,k,\ell,s_1,\ldots,t_{\ell})=0$.
\item A collection of scenarios is plausible if every vertex of the graph belongs to at most one of its scenarios.
\item We return $\mbox{MinDefect}(G,k,\ell,s_1,\ldots,t_{\ell})=\infty$ if there is no plausible collection of scenarios.
\item Instead of processing all scenario collections until finding one that allows a Hamiltonian linkage, we compare the defects $\mbox{Defect}(G,k,{\cal S})$ of all scenario collections and return the minimum one.
\item The defect with respect to one scenario collection $\cal S$ is computed recursively as
$$\mbox{Defect}(G,k,{\cal S})=|A\setminus\bigcup{\cal S}|+\sum_{i=1}^s \mbox{MinDefect}(G[Q'_i],k-1,\ell'_i,s'_i,\ldots,t'_i).$$  
\end{itemize} 
The running time is essentially the same as of Algorithm~\ref{alg:algorithm}, i.e., $O(n^{f(k,\ell)})$. By checking all possible $O(n^{2\ell})$ pairs of need-to-be-linked vertices, we can determine the maximum number of vertices that can be covered by $\ell$ disjoint paths in a $kK_1$-free graph $G$ in time $O(n^{2\ell+f(k,\ell)})$. By checking all pairs of adjacent vertices, we can determine the length of a longest cycle in time $O(n^{2+f(k,1)})$.   

\subsection{Constructing Hamiltonian linkages, paths and cycles}\label{subsec:construct}

If we are interested in the constructive versions of the problems, rather than only the decision ones, we can actually construct a maximum linkage in a similar running time. In the recursive calls we return the optimal linkages of the components $Q_i'$ instead of just the defect, and we concatenate them to an optimal linkage of entire $G$. 

Slightly more tricky is the case of highly connected input graphs. If $c_v(G)\ge g(k,\ell)$, we first construct a linkage of the input vertices via the method of~\cite{thomas2005improved}. Their proof of existence can be turned into a constructive one, the most costly part is finding a rigid separation, which can be done by brute force in $O(n^{2\ell})$ time. So altogether in $O^*(n^{2\ell})$ 
time we find a linkage of the need-to-be-linked vertices. Then we employ our proof of Theorem~\ref{thm:Ham-linked}. If the current linkage $P_1,\ldots,P_{\ell}$ does not contain all vertices, we pick a vertex $x$ uncovered by this linkage, find the paths $R_1,\ldots,R_m$ using network flow techniques in polynomial time and improve the linkage along the lines of the proof. This improving step needs to be repeated at most $n$ times.    
We conclude that the case of highly connected input graphs can be solved in $O^*(n^{2\ell})$ time. 
Since $f(k,\ell)\ge 2\ell$, the overall running time is upper bounded by the running time of the recursive case.

\subsection{Fixed parameter tractability}

From the Fixed Parameter Tractability point of view,  
our results can be restated so that {\sc Hamiltonian-$\ell$-Linkage}, {\sc Hamiltonian-$\ell$-Linkedness}, {\sc $\ell$-Path-Cover}, {\sc Maximum-$\ell$-Linkage} and {\sc Maximum-$\ell$-Path-Cover} are in the class XP when parameterized by $\ell$ and $\alpha(G)$. The natural question is the following.

\begin{problem}
Is any of these problems in FPT when parameterized by $\ell + \alpha(G)$?
\end{problem}

The problems with bounded $\ell$, in other words, {\sc Hamiltonian-Path}, {\sc Hamiltonian-Cycle}, {\sc Hamiltonian-$1$-Linkage} and {\sc Path-Cover} are in the class XP when parameterized by $\alpha(G)$. Thus we ask the following question.

\begin{problem}
Is any of these problems in FPT when parameterized by $\alpha(G)$?
\end{problem}

This is also a good place to comment on an estimate of the function $f(k,\ell)$. A careful analysis of the recursion $f(k,\ell)\le 2g(k\ell)+f(k-1,g(k,\ell)+\ell)$ shows that $f(k,\ell)=O((k+1)!\cdot \ell)$. We omit the proof, since our primary goal was just to show that $f(k,\ell)< \infty$.

\subsection{Graphs with bounded tree independence number} 

A very recent approach in algorithmic applications of width parameters of graphs is the so called \emph{tree independence number}. 
Dallard, Milanič and Štorgel~\cite{dallard2021tree} define it in the following way. If 
 $\mathcal{T} = (T, \{ X_t \}_{t\in V(T)})$ is a tree decomposition of a graph $G$, the \emph{tree independence number} of $\cal{T}$  is defined as 
$$ \alpha(\mathcal{T}) = \mathrm{max}_{t \in V(T)} (\alpha(G[X_t])$$
and the tree independence number of $G$, denote by $\mathrm{tree-}\alpha(G)$ is the minimum tree independence number $ \alpha(\mathcal{T})$ among all possible tree decompositions $\mathcal{T}$ of $G$. They prove that several optimization problems are polynomial time solvable on graphs of bounded tree independence number.

However, the hope that our results could be extended to graphs of bounded tree independence number has vanished very quickly. 
Hamiltonian circuit and Hamiltonian path are NP-complete for chordal graphs, cf.~\cite{bertossi1986hamiltonian}, i.e., for graphs of tree independence number equal to 1.

\bibliography{references}

\begin{thebibliography}{10}
\providecommand{\url}[1]{\texttt{#1}}
\providecommand{\urlprefix}{URL }
\providecommand{\doi}[1]{https://doi.org/#1}

\bibitem{bertossi1986hamiltonian}
Bertossi, A.A., Bonuccelli, M.A.: Hamiltonian circuits in interval graph
  generalizations. Information Processing Letters  \textbf{23}(4),  195--200
  (1986)

\bibitem{BKTL00}
Bodlaender, H.L., Kloks, T., Tan, R.B., van Leeuwen, J.: lambda-coloring of
  graphs. In: Reichel, H., Tison, S. (eds.) {STACS} 2000, 17th Annual Symposium
  on Theoretical Aspects of Computer Science, Lille, France, February 2000,
  Proceedings. Lecture Notes in Computer Science, vol.~1770, pp. 395--406.
  Springer (2000). \doi{10.1007/3-540-46541-3\_33},
  \url{https://doi.org/10.1007/3-540-46541-3\_33}

\bibitem{bo04}
Bodlaender, H.L., Kloks, T., Tan, R.B., Van~Leeuwen, J.: Approximations for
  $\lambda$-colorings of graphs. The Computer Journal  \textbf{47}(2),
  193--204 (2004)

\bibitem{bondy2003short}
Bondy, J.A.: Short proofs of classical theorems. Journal of Graph Theory
  \textbf{44}(3),  159--165 (2003)

\bibitem{br00}
Brandst{\"a}dt, A., Dragan, F.F., K{\"o}hler, E.: Linear time algorithms for
  {Hamiltonian} problems on (claw, net)-free graphs. SIAM Journal on Computing
  \textbf{30}(5),  1662--1677 (2000)

\bibitem{cha62}
Chang, G.J., Kuo, D.: The l(2,1)-labeling problem on graphs. SIAM Journal on
  Discrete Mathematics  \textbf{9}(2),  309--316 (1996)

\bibitem{Spirkl2019}
Chudnovsky, M., Spirkl, S., Zhong, M.: Four-coloring ${P}_6$-free graphs. In:
  Proceedings of the Thirtieth Annual ACM-SIAM Symposium on Discrete
  Algorithms. pp. 1239--1256. SIAM (2019)

\bibitem{chvatal1972hamilton}
Chv{\'a}tal, V.: On {Hamilton}'s ideals. Journal of Combinatorial Theory,
  Series B  \textbf{12}(2),  163--168 (1972)

\bibitem{chvatal1972note}
Chv{\'a}tal, V., Erd{\"o}s, P.: A note on {Hamiltonian} circuits. Discret.
  Math.  \textbf{2}(2),  111--113 (1972)

\bibitem{dallard2021tree}
Dallard, C., Milani{\v{c}}, M., {\v{S}}torgel, K.: Tree decompositions with
  bounded independence number and their algorithmic applications. arXiv
  preprint arXiv:2111.04543  (2021)

\bibitem{damaschke1989hamiltonian}
Damaschke, P.: The {Hamiltonian} circuit problem for circle graphs is
  {NP}-complete. Information Processing Letters  \textbf{32}(1), ~1--2 (1989)

\bibitem{damaschke1993paths}
Damaschke, P.: Paths in interval graphs and circular-arc graphs. Discrete
  Mathematics  \textbf{112}(1-3),  49--64 (1993)

\bibitem{damaschke1991finding}
Damaschke, P., Deogun, J.S., Kratsch, D., Steiner, G.: Finding {Hamiltonian}
  paths in cocomparability graphs using the bump number algorithm. Order
  \textbf{8}(4),  383--391 (1991)

\bibitem{dean1993computational}
Dean, A.M.: The computational complexity of deciding
  {hamiltonian}-connectedness. Congressus Numerantium pp. 209--209 (1993)

\bibitem{denley2001generalization}
Denley, T., Wu, H.: A generalization of a theorem of {Dirac}. Journal of
  combinatorial theory. Series B  \textbf{82}(2),  322--326 (2001)

\bibitem{dirac1952some}
Dirac, G.A.: Some theorems on abstract graphs. Proceedings of the London
  Mathematical Society  \textbf{3}(1),  69--81 (1952)

\bibitem{Duf1981}
Duffus, D., Gould, R., Jacobson, M.: Forbidden subgraphs and the {Hamiltonian}
  theme. In: I4th Int. Conf. on the Theory and Applications of Graphs,
  Kalamazoo, 1980. pp. 297--316. Wiley (1981)

\bibitem{FGK05}
Fiala, J., Golovach, P.A., Kratochv{\'{\i}}l, J.: Distance constrained
  labelings of graphs of bounded treewidth. In: Caires, L., Italiano, G.F.,
  Monteiro, L., Palamidessi, C., Yung, M. (eds.) Automata, Languages and
  Programming, 32nd International Colloquium, {ICALP} 2005, Lisbon, Portugal,
  July 11-15, 2005, Proceedings. Lecture Notes in Computer Science, vol.~3580,
  pp. 360--372. Springer (2005). \doi{10.1007/11523468\_30},
  \url{https://doi.org/10.1007/11523468\_30}

\bibitem{Garey:2000}
Garey, M.R., Johnson, D.S.: Computers and Intractability, A Guide to the Theory
  of NP-Completeness. W. H. Freeman and Company, New York, 22 edn. (2000)

\bibitem{garey1976planar}
Garey, M.R., Johnson, D.S., Tarjan, R.E.: The planar {Hamiltonian} circuit
  problem is {NP}-complete. SIAM Journal on Computing  \textbf{5}(4),  704--714
  (1976)

\bibitem{ge94}
Georges, J.P., Mauro, D.W., Whittlesey, M.A.: Relating path coverings to vertex
  labellings with a condition at distance two. Discrete Mathematics
  \textbf{135}(1-3),  103--111 (1994)

\bibitem{golumbic2004algorithmic}
Golumbic, M.C.: Algorithmic graph theory and perfect graphs. Elsevier (2004)

\bibitem{gould2003advances}
Gould, R.J.: Advances on the {Hamiltonian} problem -- a survey. Graphs and
  Combinatorics  \textbf{19}(1),  7--52 (2003)

\bibitem{gould2014recent}
Gould, R.J.: Recent advances on the {Hamiltonian} problem: Survey {III}. Graphs
  and Combinatorics  \textbf{30}(1),  1--46 (2014)

\bibitem{griggs1992labelling}
Griggs, J.R., Yeh, R.K.: Labelling graphs with a condition at distance 2. SIAM
  Journal on Discrete Mathematics  \textbf{5}(4),  586--595 (1992)

\bibitem{henzinger2000computing}
Henzinger, M.R., Rao, S., Gabow, H.N.: Computing vertex connectivity: new
  bounds from old techniques. Journal of Algorithms  \textbf{34}(2),  222--250
  (2000)

\bibitem{jedlivckova2024structure}
Jedli{\v{c}}kov{\'a}, N., Kratochv{\'\i}l, J.: On the structure of
  {Hamiltonian} graphs with small independence number. arXiv preprint
  arXiv:2403.03668  (2024)

\bibitem{Karp1972}
Karp, R.M.: Reducibility among Combinatorial Problems, pp. 85--103. Springer
  US, Boston, MA (1972)

\bibitem{keil1985finding}
Keil, J.M.: Finding {Hamiltonian} circuits in interval graphs. Information
  Processing Letters  \textbf{20}(4),  201--206 (1985)

\bibitem{kuvzel2012thomassen}
Ku{\v{z}}el, R., Ryj{\'a}{\v{c}}ek, Z., Vr{\'a}na, P.: Thomassen's conjecture
  implies polynomiality of 1-{Hamilton}-connectedness in line graphs. Journal
  of Graph Theory  \textbf{69}(3),  241--250 (2012)

\bibitem{mu96}
M{\"u}ller, H.: Hamiltonian circuits in chordal bipartite graphs. Discrete
  Mathematics  \textbf{156}(1-3),  291--298 (1996)

\bibitem{mu17}
Munaro, A.: On line graphs of subcubic triangle-free graphs. Discrete
  Mathematics  \textbf{340}(6),  1210--1226 (2017)

\bibitem{Mutze2016}
M{\"u}tze, T.: Proof of the middle levels conjecture. Proceedings London Math.
  Society  \textbf{112}(4),  677--713 (2016)

\bibitem{Mutze2018}
M{\"{u}}tze, T., Nummenpalo, J., Walczak, B.: Sparse kneser graphs are
  hamiltonian. In: Diakonikolas, I., Kempe, D., Henzinger, M. (eds.)
  Proceedings of the 50th Annual {ACM} {SIGACT} Symposium on Theory of
  Computing, {STOC} 2018, Los Angeles, CA, USA, June 25-29, 2018. pp. 912--919.
  {ACM} (2018). \doi{10.1145/3188745.3188834}

\bibitem{ore1960note}
Ore, O.: A note on {Hamiltonian} circuits. American Mathematical Monthly
  \textbf{67}, ~55 (1960)

\bibitem{thomas2005improved}
Thomas, R., Wollan, P.: An improved linear edge bound for graph linkages.
  European Journal of Combinatorics  \textbf{26}(3-4),  309--324 (2005)

\end{thebibliography}

\end{document}